\documentclass[10pt,a4paper]{article}
\usepackage{amsmath,amssymb,amsfonts,amsthm}
\usepackage{graphicx}
\usepackage{t1enc}
\usepackage[english]{babel}

\usepackage[latin1]{inputenc}
\usepackage{color}
\usepackage{verbatim}
\usepackage{hyperref}
\usepackage{color}
\usepackage{booktabs}
\newtheorem{theorem}{Theorem}[section]

\newtheorem{definition}[theorem]{Definition}

\theoremstyle{definition}

\usepackage{enumerate}

\newcommand{\be}{\begin{equation}}

\newcommand{\ee}{\end{equation}}

\newcommand{\Rt}{\mathbb{R}^3}

\newcommand{\RA}{\mathcal{R}_A}
\newcommand{\RC}{\mathcal{R}_C}
\newcommand{\Ho}{\mathcal{H}}

\newcommand{\R}{\mathcal{R}}

\title{ Penrose-like inequality with angular momentum for minimal surfaces.}
\author{Pablo Anglada \footnote{panglada@famaf.unc.edu.ar}\\
  Facultad de Matem\'atica, Astronom\'{i}a, F\'{i}sica y Computaci\'on \\
     Universidad Nacional de C\'ordoba, \\
Instituto de F\'{i}sica Enrique Gaviola, IFEG, CONICET,\\
  Ciudad Universitaria (5000) C\'ordoba, Argentina.}

\begin{document}
\maketitle
\begin{abstract}
In axially symmetric space-times the Penrose inequality can be strengthened to include angular momentum. 
We prove a version of this inequality for minimal surfaces, more precisely, a lower bound for the  ADM mass in terms of the area of a minimal surface, the angular momentum and a particular measure of the surface size. We consider an axially symmetric and asymptotically flat initial data, and use the monotonicity of the Geroch quasi-local energy on 2-surfaces along the inverse mean curvature flow. 
\end{abstract}

\section{Introduction}
 
In 1973 Penrose \cite{Penrose1973} proposed, using a heuristic argument, that the mass $m$ of a black hole must satisfy the relation 
\be\label{penroseineq}
m\geq \sqrt{\frac{A}{16\pi}}
\ee
where $A$ is the area of the black hole horizon. Since the original proposal by Penrose this topic have become an active area of research, and several versions of the problem have been studied, see the review articles \cite{Mars2009, Bray2003}, as well as the general approaches to it  \cite{Frauendiener:01, BrayKhuri:11, BrayKhuri:09}. Moreover under some conditions, one can strengthen the Penrose heuristic argument to include charge and angular momentum (see \cite{Dain:2012} for more details).

The Penrose argument goes as follows: assume we have collapsing matter, from the \textit{cosmic censorship conjecture} the end state of the collapse must be a black hole in equilibrium. In this situation it is expected that the matter fields lies within the event horizon and hence, according to the black hole uniqueness result \cite{Heusler:1996}, the end state is a Kerr-Newman black hole. Then the area $A_f$ of the final state black hole is given in terms of the total mass $m_f$, the angular momentum $J_f$ and the electric charge $Q_f$ of the black hole as:
\be \label{Afmjq}
A_f=4 \pi \left(2 m_f^2 - Q_f^2 +2 m_f \sqrt{m_f^2- \frac{J_f^2}{m_f^2} - Q_f^2} \right).
\ee
Take a Cauchy surface $M$ such that the collapse has already occurred, let $\varSigma$ denote the intersection of the event horizon with the Cauchy surface, and $A_{\varSigma}$ be the area of $\varSigma$. In $M$ we compute the total mass $m$, the charge $Q$, and the angular momentum $J$ at space-like infinity. From the black hole area theorem \cite{Hayward:1994} we have that the area increases with time, $A_f\geq A$, and since gravitational waves carry positive energy, the total mass measured on $M$ must be bigger than the mass of the final state black hole, $m \geq m_f$ .
Then using that the area of the final state black hole \eqref{Afmjq} is a monotonically increasing function of $m_f$ we have:
\be \label{Amjq}
A_{\varSigma} \leq A_f \leq 4 \pi \left(2 m^2 - Q_f^2 +2 m \sqrt{m^2- \frac{J_f^2}{m^2} - Q_f^2} \right).
\ee

Good progress has been made in considering the implications of \eqref{Amjq} in the
case of a charged black hole without angular momentum. In this case
different versions of an inequality relating the ADM mass, the area of the
horizon and the electric charge have been studied
\cite{Weinstein:2005,Disconzi:2012, KhuriYamada:2013, Khuri:2013, Khuri:2014,
Khuri:2015}. This results are an important test of the weak cosmic censorship
conjecture in the case of a collapsing charge. In this work we are concerned
with the case with angular momentum and, for simplicity, we do not consider
charged initial data. This case represents a very relevant physical scenario
since in general all collapsing matter is expected to have angular momentum.
Thus a generalization of the Penrose inequality to include angular momentum is
another important test that supports the cosmic censorship conjecture.

When we consider rotating collapsing matter, the initial and final values of the angular momentum does not necessarily coincide, since gravitational waves can carry angular momentum. Hence there is no simple way to relate the value of angular momentum of the initial collapsing state with the final black hole. To avoid this problem one must consider a situation in which the angular momentum is conserved along the evolution, and this is just the case we have when we assume the spacetime  is \textit{axially symmetric}. Is important to note that the axially symmetric condition is necessary to have a good definition of angular momentum in the general relativity context, we do not know how to calculate the angular momentum with out it. Thus this condition is necessary not only to have a conversed angular momentum but also to define the angular momentum it self.

Then considering axially symmetric black holes generated by an axially symmetric collapse, we have:
\be \label{mvsj}
A_{\varSigma} \leq 8 \pi \left( m^2 + m \sqrt{m^2- \frac{J^2}{m^2}} \right)
\ee

Now, we would like to bound the area $A_{\varSigma}$ in terms of geometrical quantities that can be calculated from the initial conditions. The problem is that the standard mathematical definition of black holes does not allows us to do this. The  existence of an event horizon is a global property of the casual structure of the entire spacetime, and we can not identify a black hole region without that global knowledge. Instead, the concept of trapped surfaces provides a local characterization of a black hole, and moreover they imply the existence of spacetime singularities and  (with the appropriate causal structure) the existence of event horizons (see \cite{Booth2005} for a full discussion of this topic). The boundary of regions of trapped surfaces are horizons, in particular we are going to consider future apparent horizons. A future apparent horizon $\Ho$ is a  2-surface defined by the property that all outgoing future directed null geodesics orthogonal to $\Ho$ have expansion $\Theta^+=0 $  and the 
expansion of the outgoing past null geodesics is non-negative $\Theta^-\geq 0$.

Then considering a cut $S$ of the outer trapped boundary $\Ho$ with a Cauchy surface $M$ in an axially symmetric space time, the Penrose heuristic  argument implies  (\cite{Dain:2002}, \cite{Hawking:1972}) :
\be \label{Avsmj}
A \leq 8 \pi m \left( m + \sqrt{m^2 - \frac{J^2}{m^2}} \right) 
\ee
where $A$ is the area of $S$.
This inequality only makes sense if the data satisfies $m \geq \sqrt{|J|}$. This condition has been proved for a connected outer trapped surface in \cite{Dain:2008}. Moreover in 2011 Dain and Reiris \cite{Dain:2011pi} proved that, under certain conditions, the area and angular momentum of a axially symmetric closed apparent horizon satisfies the local inequality $A \geq 8 \pi|J|$. This result has also been extended to other cases in \cite{Acena:2010ws, Clement:2012vb, Clement:2015fqa}.
Thus inequality \eqref{Avsmj} is equivalent to (see \cite{Dain:2014Geo}, \cite{Mars2009}):
\be \label{penJ}
m^2 \geq \frac{A }{16 \pi} + \frac{4 \pi J^2}{A}
\ee
This version of the Penrose inequality admits a rigidity case which states that the equality can only occur for the Kerr black hole.

In this paper we study the validity of inequality \eqref{penJ} and the possibility of proving some weak version of it. This is one of the open problems in the geometrical 
inequalities area (see \cite{Dain:2014Geo,Mars2009}), and not much progress has been achieved in obtaining a formal proof of this inequality.

Inequality \eqref{penJ} establishes a relation between the mass measured at space-like infinity and quasi-local quantities of an apparent horizon, its area and angular momentum. Then, to address the problem one has to select a hypersurface $M$ of the spacetime to connect the horizon with infinity \cite{Malec:2002ki}. Hence we have to study an \textit{initial data} in accordance with the spacetime, and calculate the physical parameters involved in \eqref{penJ} in terms of this initial data, note that in particular the total mass will be the ADM mass $m_{ADM}$ \cite{ADM}.
An initial data set $(M,\bar g, K; \mu ,j^i)$  is given by a 3-manifold $M$ with positive definite metric $\bar g$ and extrinsic curvature $K$, together with an energy density  $\mu$ and a matter current $j^i$. This set must satisfy the constraint equations 
\begin{align}
 \label{const1}
   \bar D_j   K^{i j} -  \bar D^i   k= -8\pi j^i,\\
 \label{const2}
   \bar R -  K_{i j}   K^{i j}+  k^2=16\pi \mu,
\end{align}
where ${\bar D}$ and $\bar R$ are the Levi-Civita connection and the curvature scalar associated with $  \bar g$, and $k=\mbox{tr}K$. 
We assume the data is \textit{asymptotically flat} and that the matter fields satisfy the \textit{Dominant Energy Condition} (DEC) , $\mu \geq |j|$,  moreover we suppose $M$ has a boundary given by a 2-surface $\partial M$ which is a future apparent horizon. These are probably the most general assumptions one needs to make to be able to treat the problem.

In this work we focus on the case where the apparent horizon is a minimal surface, $\Theta^+=\Theta^-=0$. This particular assumption for the apparent horizon is usually associated with the special case in which the initial data is time symmetric, usually known as the Riemannian case. This terminology is not precise enough (\cite{Mars2009} \cite{Bray2003}): in the case where the apparent horizon is a minimal surface, the time symmetric condition is not essential. The inequality \eqref{penroseineq} holds as long as the scalar curvature is non-negative \cite{Huisken1997} \cite{Huisken2001}, and this condition can also be satisfied taking weaker assumptions, for example considering only maximal initial data sets. Note that in the time symmetric case the angular momentum of the initial data is zero, thus we need some weaker conditions if we want to include the effects of rotation on the Pensore inequality. 

Following \cite{Anglada2017} we are going to use a certain functional proposed by Geroch \cite{Geroch1973}, the Geroch energy $E_G$, which is monotonic under a  smooth \textit{inverse mean curvature flow} (IMCF). 
This functional has the interesting property that tends to the ADM Mass of $M$ at infinity and is equal to $\sqrt{\frac{A}{16 \pi}}$ for a minimal surface. This argument was used by Huisken and Ilmanem in their important work \cite{Huisken2001}, where they removed the smoothness assumption of the IMCF using a weak formulation. See \cite{Huiskenevol} \cite{Szabados04} for a review of the basic properties of the inverse mean curvature flow (IMCF) and the Geroch energy. A solution of the IMCF is a smooth family of hypersurfaces  $S_t:=x(S,t)$ on $M$, with $x:S\times[0,\tau]\to M$ satisfying the evolution equation
\be \label{eqIMCF}
\frac{\partial x}{\partial t}=\frac{\nu}{H}
\ee
where $t\in[0,\tau]$, $H>0$ is the mean curvature of the 2-surface $S_t$ at $x$ and $\nu$ is the outward unit normal to $S_t$.
The Geroch energy is defined for each surface of the flow as follows:
\be \label{Gmass}
E_G(S_t):=\frac{A_t^{1/2}}{(16\pi)^{3/2}}\left(16\pi-\int_{S_t}H^2dS\right)
\ee
where $A_t$ and $dS$ are the area and the area element of $S_t$ respectively, and the time derivate of $E_G(S_t)$ along the flow satisfies
\be\label{evolE}
\frac{d}{dt}E_G \geq \frac{A_t^{1/2}}{(16\pi)^{3/2}}\int_{S_t} \bar R dS.
\ee
Using these tools we prove an inequality in the spirit of \eqref{penJ}, that is to say, a lower bound for the total mass of the initial data in terms of the area of the minimal surface and a quotient between the angular momentum and certain measures of size. For this problem the monotonicity of the Geroch energy is not enough to include the rotational contribution, we need to relate the behavior of the energy along the IMCF with the angular momentum of the initial data. We carry out this by relating the surface integral of the scalar curvature in \eqref{evolE} with the angular momentum and a particular measure of size that depends on the behavior of the axial Killing vector along the IMCF \cite{Anglada2017} .

\section{Background}

Following \cite{Malec:2002ki} we consider an asymptotically flat and axially symmetric initial data with boundary $(M, \partial M,\bar g, K; \mu, j^i)$, such that the boundary $\partial M$ is connected and compact, and the matter fields satisfy the DEC. We assume $\partial M$ is a minimal surface, and that there are no other trapped surfaces on $M$. With these assumptions $M$ is an \textit{exterior region}, has the topology $\Rt$ minus a ball and its boundary is an area-minimizing 2-surface \cite{Huisken2001}. Assume there exists a smooth inverse mean curvature flow (IMCF) of surfaces $S_t$ starting from $S_{t_0}=\partial M$ and having spherical topology. With these assumptions one can write the metric $\bar g$ in the form:
\be\label{barg}
ds^2_{\bar g}=\frac{dt^2}{H^2} + g_{ij}dx^i dx^j
\ee
where $g_{ij}$ and $(x^1, x^2)$ are the induced metric and the coordinates on $S_t$ respectively.
Note that when considering the IMCF in axially symmetric initial data, 
the IMCF equation \eqref{eqIMCF} preserves axial symmetry. That is, if one starts the flow with an axially 
symmetric initial surface, there is no mechanism that could make the normal to each subsequent surface to have a component along the axial Killing vector field $\eta^i$ associated to the axial symmetry. Due to this 
observation, from now on, when we discuss the IMCF flow, we always consider it consisting of axially symmetric surfaces $S_t$. Then for each surface of the flow we can we define orthogonal coordinates $\theta, \varphi$ such that $\eta^i=\frac{\partial}{\partial \varphi}^i$. One can always choose this for axially symmetric 2-surfaces that are diffeomorphic to $S^2$ , see for example
\cite{Dain:2011pi}. Hence we have:
\be \label{gSt}
  ds^2_g=\Psi^4 d\theta^2+ \eta d\varphi^2
\ee
where $\eta=g_{ij}\eta^i \eta^j$ is the norm of the axial Killing vector.
Moreover in this setting the extrinsic curvature can be decomposed \cite{Malec:2002ki}:
\be\label{K}
K_{ij}=z\nu_i \nu_j + \nu_i s_j + s_i \nu_j  + g^k_i g^l_j \chi_{lk} + \frac{q}{2}g_{ij}
\ee
where q is the trace with respect to $g_{ij}$ of $K$, $q=K_{ij}g^{ij}$ and 
\be\label{Kparts}
z=K_{ij}\nu^i \nu^j \qquad s_i=g^j_i K_{jl} \nu^l \qquad \chi_{ij}=g^l_i g^n_j K_{ln} - \frac{q}{2} g_{ij}
\ee
then the trace of the extrinsic curvature takes the form $k=\mbox{tr}(K)=z+q$ and its norm is
\be\label{KK}
K_{ij}K^{ij}=z^2+2s_i s^i + \chi_{ij} \chi^{ij} + \frac{q^2}{2}.
\ee
In this context the null expansions $\Theta^+, \Theta^-$ of $S_t$ are given by $ \Theta^+|_{S_t}=H+q$ and  $\Theta^-|_{S_t}=H-q$. Then if $M$ has no other trapped surface than $\partial M$, the expansions satisfy 
\be\label{NTS}
(\Theta^+\Theta^-)|_{S_t}=H^2-q^2>0 \quad \forall t \neq t_0,
\ee
Note that in particular this implies that there are no minimal surfaces on $M$, except for $\partial M$, and this is a necessary condition to have a smooth IMCF \cite{Huisken2001}.

Now in order to have a non-negativity scalar curvature we need to assume some special conditions for the extrinsic curvature. The fist possible choice, the most usual, is to consider that the initial data is maximal, $k=0 $. Note that in this case, the minimal surface $\partial M$ is an apparent horizon only if the extrinsic curvature also satisfies $q|_{\partial M}=z|_{\partial M}=0$.
Another possible condition to have a non-negative scalar curvature is to take $K$ such that for every surface of the flow $q=0$, this assures that $\partial M$ is an apparent horizon. The non-negative of $\bar R$ can be seen using the previous decomposition of $K$, from equation \eqref{const2} we have:
\begin{equation}
\label{barR}
\bar  R= 16\pi \mu + 2s_i s^i + \chi_{ij} \chi^{ij} -  \frac{q}{2}(q + 4z) 
\end{equation}
and thus if $q=0$, $\bar R$ is non-negative (we have assumed the DEC).

The physical and geometrical quantities we are interested in are the ADM mass $m_{ADM}$ and the Komar angular momentum $J(S_t)$: 
\be\label{angmom}
 J(S_t)=\frac{1}{8\pi}\int_{S_t} K_{ij} \eta^i \nu^jdS,
\ee
where we use that $\bar g_{ij} \nu^i \eta^j=0$. We will also consider the areal and circumferential radii of a surface $S_t$ in $M$:
\be\label{size}
\RA(S_t):=\sqrt{\frac{A_t}{4\pi}},\qquad \RC(S_t):=\frac{\mathcal C (S_t)}{2\pi}
\ee
$A_t$ is the area of $S_t$ and $\mathcal C(S_t)$ is the length of the greatest axially symmetric circle in $S_t$.

\section{Main Result}

Together with the usual definition of size \eqref{size}, we define another measure of size of a surface $S_t$ based on the behavior of the norm of the Killing vector along the IMCF from $S_t$ to infinity :
\begin{definition}{ $\R(S_t)$}
\be\label{R}
\frac{1}{\R(S_t)^2} := A_t^{1/2} \int_{t}^\infty \frac{A_{t'}^{1/2}}{\int_{S_{t'}}  \eta dS }dt'
\ee
\end{definition}
From the asymptotic behavior of the flow one can see that the integral in \eqref{R} is convergent provided the flow remains smooth, thus $\R$ is positive and well defined. Moreover, depending on the properties of the flow this measure can in some cases be related with the usual measures of size of a surface. 
To study the properties of $\R$ we start with the evolution equation for the norm of the Killing vector $\eta$ along the flow:
\be\label{evolg}
\frac{\partial}{\partial t} \eta=\frac{2 \eta \lambda_\varphi}{H}
\ee
where we define $\lambda_\theta, \lambda_\varphi $ as the principal curvatures of $S_t$, and in particular $\lambda_\varphi$ is the principal curvature in the direction of the Killing vector. This equation comes from the evolution equation of the metric $g$ (see \cite{Anglada2017}, \cite{Huiskenevol} for more details) and is valid only for an axially symmetric IMCF. Then the behavior of the integrand in \eqref{R} is given by:
\be\label{Rdt}
\frac{d}{dt} \left( \frac{A_t^{1/2}}{\int_{S_t} \eta dS}  \right)= - \frac{5}{2} \frac{A_t^{1/2}}{\int_{S_t}\eta dS }  + \frac{A_t^{1/2}}{ \left( \int_{S_t}\eta dS \right)^2 } \int_{S_t}\frac{2 \eta \lambda_\theta}{H}dS
\ee
where we have used that the area element $dS$ of 
$S_t$ satisfies $\frac{\partial}{\partial t}(dS)= dS$ and that $H=\lambda_\theta + \lambda_\varphi$.

Now  depending on the behavior of the flow we can estimate the value of $\R$. 
The most favourable situation, and by far the most frequent one, is to have a spherical IMCF. In this case the principal curvatures are $\lambda_\theta = \lambda_\varphi = \frac{H}{2}$, then
\be
\frac{d}{dt} \left(\frac{A_t^{1/2}}{\int_{S_t} \eta dS}  \right)= - \frac{3}{2} \frac{A_t^{1/2}}{\int_{S_t} \eta dS}
\ee
and hence we have an exact expression for $\R$:
\be
\R^2(S_t)=\frac{3}{2}\frac{\int_{S_t} \eta dS}{A_t} = \RA^2(S_t)
\ee
where the last equality comes from the fact that for spherical surfaces one can write $\eta=\RA^2 \sin^2(\theta)$, and then $\int_{S_t} \eta dS= \frac{2}{3} 4\pi \RA^4 (S_t) $.
For weaker conditions on the behavior of the flow we can not obtain an exact expression for $\R$ in terms of the usual measures, but we do obtain a bound for it. Let's assume that we have a convex flow, then $\lambda_\theta , \lambda_\varphi >0$ and hence:
\be
\frac{d}{dt} \left( \frac{A_t^{1/2}}{\int_{S_t} \eta dS}  \right) \geq - \frac{5}{2} \frac{A_t^{1/2}}{\int_{S_t}\eta dS }  
\ee
thus for this case:
\be
\R^2(S_t) \leq \frac{5}{2}\frac{\int_{S_t} \eta dS}{A_t} \leq \frac{5}{2} \RC^2(S_t).
\ee
Moreover we can also get a similar relation between $\R$ and $\RC$ assuming an even weaker condition: that the surfaces of the flow are not far from being convex. In particular assuming that $h_{ij}\geq -\frac{H}{2} g_{ij}$ we have
\be
\R^2(S_t) \leq \frac{7}{2}\frac{\int_{S_t} \eta dS}{A_t} \leq \frac{7}{2} \RC^2(S_t).
\ee
It is important to note that from the asymptotic behavior of the flow there must exist some first surface such that the flow satisfies these properties from then on, even without assuming any especial property for the IMCF.  Moreover 
the weaker the condition we use, the closer the surface to $\partial M$.

Using this definition of size, and the previous tools we prove the following theorem.


\begin{theorem}
\label{theo1}
Let $(M,\partial M, \bar g, K)$
be a vacuum, asymptotically flat, and axially symmetric initial data, such that $\partial M$ is a compact and connected minimal surface.
Suppose the data has no other trapped surfaces and assume there exists a smooth IMCF of surfaces $S_t$ starting from $\partial M$ and having spherical topology.
Then if the data satisfies either:
\begin{enumerate}[a)]
        \item \label{a} the initial data is maximal: $k=0$, or 
        \item \label{b} for each surface $S_t$ the trace of $K$ with respect to $g_{ij}$ satisfies: $q=0$,
\end{enumerate}
then:
\be\label{mainineq}
m_{ADM}\geq \frac{\RA}{2} + \frac{ J^2}{\RA \R^2}
\ee
where $J$ and $\RA$ are the angular momentum and the areal radii of $\partial M$ respectively, and $\R=\R(\partial M)$ is defined by \eqref{R}.
\end{theorem}

\begin{proof}

From \cite{Anglada2017} we know that one can include the rotational contribution to the energy in the derivate of the Geroch energy. The idea is to relate the angular momentum of each surface of the IMCF with the surface integral of the norm of the extrinsic curvature, and use this to obtain a lower bound to the surface integral of the scalar curvature.

Noting that from \eqref{K} $K_{ij} \eta^i \nu^j=s_i \eta^i$, we can improve the bound obtained in \cite{Anglada2017} in the following way:
\be\label{angmomineq}
\begin{split}
 J_t^2&=\frac{1}{(8\pi)^2} \left( \int_{S_t} s_i \eta^i dS\right)^2 \leq \frac{1}{(8\pi)^2} \left( \int_{S_t} |s_i\eta^i| dS\right)^2 \\
 & \leq \frac{1}{(8\pi)^2} \left( \int_{S_t} |s_i| |\eta^i| dS\right)^2  \leq \frac{1}{(8\pi)^2} \int_{S_t} |s_i|^2 dS  \int_{S_t} |\eta^i|^2 dS  \\
 &=\frac{1}{(8\pi)^2} \int_{S_t} s_is^i dS  \int_{S_t} \eta dS
\end{split}
\ee
where in the fourth step we have used the H\"older inequality. Then we have:
\be
\int_{S_t} s_i s^i dS \geq \frac{(8\pi)^2 J_t^2}{\int_{S_t}\eta dS}
\ee
thus assuming either conditions \eqref{a} or \eqref{b} for $K$, and using \eqref{barR} we obtain the desired bound:
\be
\label{RvsJ}
\int_{S_t} \bar R dS \geq 2\int_{S_t} s_i s^i dS \geq  2\frac{(8\pi)^2 J_t^2}{\int_{S_t}\eta dS}
\ee
Hence from \eqref{evolE} and \eqref{RvsJ}, and using that the angular momentum is conserved along the flow $J(S_t)=J(\partial M)=J$ we have:
\be\label{eqvolmm}
\frac{d}{dt}E_G \geq  \sqrt{4\pi} J^2 \frac{A_t^{1/2}}{\int_{S_t}\eta dS}
\ee

Then integrating \eqref{eqvolmm} from $\partial M$ to infinity, and using the relation between the Geroch energy and the ADM mass we obtain:
\be\label{evol4}
m_{ADM}\geq\lim_{t \to \infty} E_G(S_t) \geq E_G(S_0) + \sqrt{4\pi}J^2\int_0^\infty \frac{A_t^{1/2}}{\int_{S_t}  \eta dS}dt.
\ee 
Finally we write this expression in terms of $\R$ and the areal radii of $\partial M$
and obtain \eqref{mainineq}.
\end{proof}
\vspace{0.5cm}


In case we have a non-vacuum initial data, the inequality can be extended 
using similar techniques, provided that the matter field satisfies the dominant energy condition and that both the matter density and the matter current have compact support.
Note that in this case the angular momentum of a surface $S_t$ is
\be
 J(S_t)=\frac{1}{8\pi}\int_{S_t} K_{ij} \eta^i \nu^jdS=  J(\partial M) - \int_{V(S_t)} j_{i} \eta^i dv,
\ee
where $V(S_t)$ is the region enclosed between $\partial M$ and $S_t$. Thus the conservation of the angular momentum along the flow  is only satisfied when the surfaces $S_t$ are outside the compact support of $j_i\eta^i$. Then under the same conditions of theorem \eqref{theo1} we obtain the following extension for non-vacuum initial data.


\begin{theorem}\label{theorem2}
Let $(M,\partial M, \bar g, K; \mu, j^i)$ be an initial data satisfying the same conditions of theorem  \ref{theo1}.
Assume the matter fields satisfy the dominant energy condition and have compact support, and let $T$ such that for all $t\geq T$ the matter density and the matter current have compact support inside $S_t$, then:
\be\label{theo2}
m_{ADM}\geq m_{T} +  \frac{\RA}{2} + \frac{ J^2}{\RA(T) \R^2(T)}
\ee
where $J$ is the total angular momentum of the data, $\RA$ and $\RA(T)$ are the areal raddi of $\partial M$ and $S_T$ respectively, $\R(T)=\R(S_T)$ is defined  by \eqref{R}, and 
\be\label{mTbh}
m_{T}:=\frac{1}{16\pi}\int_{\RA}^{\RA(T)}d\xi\int_{S_\xi}\bar R dS 
\ee
where $\xi$ stands for the  areal radius coordinate.
\end{theorem}

\vspace{1cm}

\textbf{Remarks} 

\vspace{0.5cm}

Inequalities \eqref{mainineq} and \eqref{theo2} are global in nature as they involve the ADM mass, are linear in the ADM mass and quadratic in the  angular momentum as a result of the linear dependence of $dE_G/dt$ with the scalar curvature $\bar R$.
The first terms in \eqref{mainineq} and \eqref{theo2}, $\frac{\RA}{2}$ and $m_{T} + \frac{\RA}{2}$, are quasi-local measures of energy of the region enclosed by $S_T$, and the term $\frac{ J^2}{\RA(T) \R^2(T)}$ is a measure of the rotation energy. This relation between the total energy and the sum of local measures of energy and rotational energy is in accordance with the Newtonian limit (see \cite{Anglada2017}). From Newtonian theory one may write the total energy of a compact region as the sum of two terms, the first including the gravitational and internal energies, denoted as $E_0$, and the second, the rotational energy
$$E\approx E_0+\frac{J^2}{2I}.$$
where $I$ is the Newtonian moment of inertia of the region. Hence
we may interpret the product $\RA(S_T) \R^2(S_T)$ as an upper bound for the Newtonian moment of inertia associated to the region  enclosed by $S_T$.
Then an IMCF, starting from some initial surface in accordance with the symmetries of the initial data, provides a foliation by topological spheres adapted to the geometry of the initial data. This particular foliation, though arbitrary, gives us a privileged coordinate system in the sense that it enables us to recover the usual relations between energy and measures of size of the Newtonian limit.

\vspace{0.5cm}
The notion of size we use, $\R$, albeit apparently artificial at first sight, comes 
from the particular method we use to relate the angular momentum with the ADM
mass, and it gives a good measure of how different the IMCF is from a spherical
one. The behavior of the norm of the Killing vector along the IMCF appears to be
a necessary ingredient if we want to consider some measure of rotational energy
as a function of the object's size. Moreover the measures based on the norm of
the Killing vector have been found to give an appropriate description of size of a
region when describing both regular objects and black holes with angular
momentum \cite{Anglada2017}, \cite{Reiris:2014tva}, \cite{Reiris:2013jaa},
\cite{Dain:2014}.

\vspace{0.5cm}
We do not know if it is possible to find a relation between $\R$ with the usual measures of size; without assuming some condition on the IMCF, the measure is strongly dependent on the flow. This is clearly one of the main open questions we want to address next.
Assuming particular properties for the IMCF we can write \eqref{mainineq} in terms of the usual measures of size.
For a spherical IMCF: 
\be
m_{ADM}\geq \frac{\RA}{2} + \frac{ J^2}{\RA^3}.
\ee
For a convex IMCF: 
\be
m_{ADM}\geq \frac{\RA}{2} + \frac{2}{5}\frac{ J^2}{\RA \RC^2}.
\ee
For a IMCF satisfying  $h_{ij}\geq -\frac{H}{2} g_{ij}$ (i.e. not far from being convex): 
\be
m_{ADM}\geq \frac{\RA}{2} + \frac{2}{7}\frac{ J^2}{\RA \RC^2}. 
\ee

\vspace{0.5cm}

Inequality \eqref{mainineq} can be also written in the form of \eqref{penJ}. First note that we can drop the positive term involving the angular momentum in \eqref{mainineq} and get $m_{ADM} \geq \frac{\RA}{2}$, then multiplying \eqref{mainineq} by $\RA/2$ and using this we have
\be\label{mainineqV2}
m_{ADM}^2 \geq \frac{A}{16 \pi} + \frac{ J^2}{2 \R^2},
\ee
note that this inequality is weaker than \eqref{mainineq}. In the case we have a spherical IMCF this implies the validity of a weaker version of \eqref{penJ}:
\be
m_{ADM}^2 \geq \frac{A}{16 \pi} + \frac{ 2\pi J^2}{ A} .
\ee

\vspace{0.5cm}

Inequality \eqref{theo2} has also meaning in the case we do not have a black hole but a regular object, like a neutron star. That is to say we can consider a regular initial data without trapped surfaces. This problem has been studied in \cite{Anglada2017}, but with the arguments presented in this work we can improve the previous result:
\be
m_{ADM}\geq m_{T} + \frac{ J^2}{\RA(T) \R^2(T)}
\ee
where, in this case, we start the flow from a point in the axis of symmetry inside the object, and the quasi-local mass is
\be \label{mT}
m_{T}:=\frac{1}{16\pi}\int_{0}^{\RA(T)}d\xi\int_{S_\xi}\bar R dS .
\ee

\vspace{0.5cm}

The condition of $M$ being an exterior region is not essential for the proof.
Instead, we can assume the weaker condition that $M$ has no minimal surfaces.
Inequality \eqref{mainineq} remains valid for initial data sets with no minimal
surfaces other than $\partial M$.
In this case the initial data could have a trapped region and $\partial M$ being
no longer an apparent horizon. Moreover, inequality \eqref{theo2} is also valid
with this weaker condition, and in this case not only the minimal surface could
be inside the trapped region but also the the matter fields.
\vspace{0.5cm}

Inequality \eqref{mainineq} is also valid for non-vacuum initial data, provided that the matter fields satisfy the DEC and that $j_i\eta^i=0$ everywhere in $M$\footnote{I thank an anonimous referee for making this observation.}. Assuming the DEC we assure that the Geroch energy remains monotonic for non-vacuum initial data. Condition $j_i\eta^i=0$ assures that the angular momentum is preserved along the flow $J(S_t)=J$, and that there is no contribution to $J$ coming from the matter fields  $J=J(\partial M)$. 
\vspace{0.5cm}

In \eqref{theo2} we are assuming that the matter fields satisfy the DEC, in fact we only need the weak energy condition for the proof, but we consider the DEC is the more appropriate condition for studying this type of inequalities with matter fields.
In \ref{theorem2} we are interested in considering the situation in which the black hole is surrounded by a compact rotating object. The assumption of $\mu$ having compact support is only needed to have a good definition of the compact object we are interested in, and also its quasi-local mass $m_T$. One can prove the same relation between total mass and total angular momentum without it.
The condition of $j^i$ having compact support can also be relaxed, we can assume a weaker condition: that the current density in the direction of the Killing vector, $j_i\eta^i$, have compact support. Note that this last weaker condition is a necessary condition to obtain our result, the conservation of the angular momentum along the flow is only satisfied when the surfaces $S_t$ are outside the compact support of $j_i\eta^i$.
\vspace{0.5cm}

If the obtained inequalities have a rigidity statement, the equality must be achieved only for the t=0 slice of a Kerr black hole. Unfortunately it is not easy to test the result in a slice of a Kerr black hole. To test the theorems one needs to calculate a smooth solution of the IMCF in that slice, and we do not have an analytical expression for it, at least not yet. One possibility is to continue looking for an analytical solution, but the complexity of the equation of IMCF in this slice do not offer much hope. The other possibility is to look for a numerical solution of the IMCF to test the rigidity of the theorem. The numerical approach probably will give us the easiest scenario of study a rigidity statement and thus is one of the points we want to study next.

\vspace{0.5cm}

The assumption of existence of a smooth solution of the IMCF seems somewhat restrictive. In this work, however, we are more concerned with presenting a clear relation between physical parameters than with discussing the subtleties
involved in the IMCF. From the Huisken and Ilmanem work \cite{Huisken2001} we know that there exists a weak solution of the flow for the initial data we are considering. Some parts of our derivation do not require smoothness, and can be repeated using the weak level set version of the flow defined by the mentioned authors. Other parts of our proof strongly depend on having a smooth foliation of the initial data. In this sense, a posible future work would be to adapt the particular method we use to address the problem for weak solutions of the IMCF. The principal obstacle to do this lies in the particular measure of size we define. This measure is based on having a smooth foliation by topological spheres, and depends on the behavior of the flow that defines such foliation. In our proof we use a particular form of the metric of the initial data. This form depends on having a smooth solution of the flow. Thus, all steps in the proof relying on the explicit form of the metric would need revision.

\section*{Acknowledgments}
I thank Maria E. Gabach-Clement and Omar E. Ortiz for enlightening discussions and for their encouraging support.
I would like also to thank Natacha Altamirano for providing me some notes of her authorship about the problem.
This work was supported by grants from CONICET and
SECyT, UNC.

\end{document}